\documentclass[submission,copyright,creativecommons]{eptcs}
\providecommand{\event}{GandALF 2018} % Name of the event you are submitting to
\usepackage{breakurl}             % Not needed if you use pdflatex only.
\usepackage{underscore}           % Only needed if you use pdflatex.
\usepackage[english]{babel}
\usepackage{graphicx}
\usepackage{amsmath}
\usepackage{amssymb}
\usepackage{amsthm}
\usepackage{url,hyperref}
\usepackage{wrapfig}
\usepackage[scr]{rsfso}
\usepackage{tikz}  %[border=5pt]
\usetikzlibrary{positioning}
\usetikzlibrary{decorations.pathmorphing,shapes}
\usepackage[colorinlistoftodos,bordercolor=orange,backgroundcolor=orange!20,linecolor=orange,textsize=scriptsize]{todonotes}
\usepackage[cal=cm]{mathalfa}

\newenvironment{lemma-repeat}[1]{\begin{trivlist}
\item[\hspace{\labelsep}{\bf\noindent Lemma~\ref{#1} }]}%
{\end{trivlist}}
\newenvironment{theorem-repeat}[1]{\begin{trivlist}
\item[\hspace{\labelsep}{\bf\noindent Theorem~\ref{#1} }]}%
{\end{trivlist}}
\newenvironment{proposition-repeat}[1]{\begin{trivlist}
\item[\hspace{\labelsep}{\bf\noindent Proposition~\ref{#1} }]}%
{\end{trivlist}}

\def\cA{{\cal A}}

\def\cF{{\cal F}}

\def\cI{{\cal I}}

\def\cK{{\cal K}}

\def\cT{{\cal T}}
\def\cV{{\cal V}}
\def\cI{{\cal I}}

\def\Gz{\cI}

\newcommand{\I}[1]{\mathit{#1}}

\newcommand{\la}{\langle}
\newcommand{\ra}{\rangle}
\newcommand{\Lcal}{\mathcal{L}}

\newcommand{\pre}{\textsf{pre}}

\newcommand{\Values}{\mathcal{V}}

\newcommand{\comment}[1]{}
\newtheorem{theorem}{Theorem}
\newtheorem{lemma}{Lemma}
\newtheorem{proposition}{Proposition}
\newtheorem{corollary}{Corollary}
\newtheorem{definition}{Definition}
\newtheorem{example}{Example}

\newcommand{\view}{{\sf view}}
\newcommand{\aview}{{\sf Aview}}

\title{A Simplicial Complex Model for Dynamic Epistemic Logic
        to study  Distributed Task Computability}
\author{\'Eric Goubault \qquad\qquad J\'er\'emy Ledent
\institute{LIX, \'Ecole Polytechnique\\ Palaiseau, France}
\email{\{goubault,jeremy.ledent\}@lix.polytechnique.fr}
\and
Sergio Rajsbaum
\institute{Instituto de Matem\'aticas, UNAM\\ Ciudad Universitaria Mexico 04510, Mexico}
\email{rajsbaum@im.unam.mx}
}
\def\titlerunning{A Simplicial Complex Model for Dynamic Epistemic Logic}
\def\authorrunning{\'E. Goubault, J. Ledent \& S. Rajsbaum}

\begin{document}
\maketitle

  \begin{abstract}
The usual epistemic $\mathbf{S5_n}$ model for  a multi-agent system is based on a Kripke frame, which is a graph whose edges are labeled with  agents that do not distinguish between two states. We propose  to uncover the higher dimensional information implicit in this structure,  by considering a dual, \emph{simplicial complex model.} 
We use dynamic epistemic logic (DEL) to study how an epistemic simplicial complex model changes after a set of agents communicate with each other. We concentrate on an action model  that represents the so-called 
\emph{immediate snapshot} communication patterns
of asynchronous agents, because it is central to distributed computability (but our setting works for other 
communication patterns). 
There are topological invariants preserved from the initial epistemic complex to the one after the action model is applied, which determine the knowledge that the agents gain after  communication.
Finally, we describe how a distributed task specification can be modeled as a DEL action model, and
show that the topological invariants determine 
whether the task is solvable. We thus provide a bridge between DEL and the 
topological theory of distributed computability, which studies task solvability
in a shared memory or message passing architecture.
   \end{abstract}

%  \begin{keyword}
%  Dynamic epistemic logic, Distributed computing, Simplicial complexes
%  \end{keyword}

%%%%%%%%%%%%%%%%%%%%%%%%%%%%%%%%%%%%%%%%%%%%%%%%%%%

% !TEX root =  ./main.tex
\section{Introduction}

The usual epistemic logic model for a multi-agent system is based on a Kripke frame, which is a graph whose edges are labeled with  agents that do not distinguish between two states. %\comment{
A Kripke $\mathbf{S5_n}$ model represents the knowledge of the agents about a given situation.
Our first goal is to expose the topological information implicit in a Kripke model, replacing it by its dual, 
a \emph{simplicial complex} model.
We prove that these simplicial models are very closely related to the usual Kripke models:
there is an equivalence of categories between the two structures.
Thus, simplicial models retain the nice properties of Kripke models, such as
soundness and completeness w.r.t.\ (a slightly modified version of) the logic $\mathbf{S5_n}$.

To explain the interest of this duality, we extend it to a dynamic setting. We found that in this context, a very natural
setting is dynamic epistemic logic (DEL)~\cite{sep-dynamic-epistemic,DEL:2007} with  \emph{action models}~\cite{baltagMS:98}. 
We extend the duality to this setting by defining a simplicial version
of action models and a corresponding product update operator.
Thus, the product update of an initial simplicial model $\cI$ and an
action model $\cA$ yields a simplicial model~$\cI[\cA]$.
%By this means we establish a connection with the theory of distributed computability via  topology~\cite{HerlihyKR:2013}.
 The possible patterns of communication permitted by the action
model determine the topological invariants of $\cI$ that are preserved in $\cI[\cA]$.

%It is used to   study how  knowledge changes when  communication events happen in a multi-agent system~\cite{sep-dynamic-epistemic,DEL:2007}.
%A very general version of DEL extends 
%ordinary epistemic logic by including \emph{action models}~\cite{baltagMS:98}.
%Given an initial multi-agent Kripke model,  a product update operator  defines a new Kripke model that results as a 
%  consequence of executing actions on the initial Kripke model.
  %}%comment

We apply our framework to study fault-tolerant distributed  computability, because its intimate relation
to topology is well understood~\cite{HerlihyKR:2013}.
Also,  DEL  has applications to numerous research areas, but to the best of our knowledge it has not been
  used to study fault-tolerant distributed computing systems.
We define a particular action model of interest, the \emph{immediate snapshot}
action model, which is well known in distributed computing because it fully preserves the topology of the initial complex.
This model corresponds to wait-free asynchronous processes operating on a shared
memory, which means that the processes run at an arbitrary speed, independent
from the others, and are not
allowed to wait for events to happen in other processes.
%These correspond to the wait-free case, where any interleaving of the operations of asynchronous\footnote{Processes are said to be asynchronous in distributed computing if each one runs at an arbitrary speed, independent of the others. Wait-free means that the code that a process is running does not include operations to wait for events to happen in other processes.} agents can happen.
 
%We also consider a simplical complex corresponding to the dual of an action model.
%The action model preserves topological invariants about the initial complex model $I$  after
%the communication actions have taken place. 
%How much of the topology is preserved, depends on $t$, the number of processes that may fail.
%We study here a class of action models that \emph{fully} preserve the topology of the initial complex,
%that  correspond to the so
%called \emph{wait-free} case, where any number of asynchronous processes (meaning that each one runs at an arbitrary speed, independent of the others)
%may crash. %, $t=n$,
%%Wand hence, in a precise sense, wait-free action models yield the \emph{least} information to the agents.

Another goal is to show how DEL can be used to 
specify a distributed~\emph{task}.
A task is the equivalent of a function in distributed computability~\cite{2004dcAW}.
Agents start with an input value, and after communicating with the others,  produce an output value.
The task defines the possible inputs to the agents, and for each 
set of inputs, it specifies %by an action model, % relation $\Delta$, 
the set of outputs that the agents may produce.
An important example is the \emph{consensus} task, where all the agents must
agree on one of their input values. We use DEL in a novel way,
%, and more generally, the \emph{$k$-set agreement task,} where processes must agree on at most $k$ of their input values.
%A main result in distributed computing is 
%that $n+1$ processes can solve $k$-set agreement if and only if at most $t$ processes can crash, for $t<k$
%~\cite{1993GeneralizedFLPImposibility_BG,1999TopologicalStructureAsynchronous_HS,SaksZ00}.
%In particular, when $k=1$, processes can solve \emph{consensus} only if no process can crash~\cite{FischerLP85}.
\begin{wrapfigure}[9]{r}{-0.4\textwidth}%   \centering
%\begin{wrapfigure}{r}{-0.4\textwidth}%   \centering
 \centering
\begin{tikzpicture}
  \node (s) {$\cI[\cA]$};
  \node (xy) [below=2 of s] {$\cI[\cT]$}; %\Delta \subseteq \cI \times \cO$};
  \node (x) [left=of xy] {$\cI$};
  \draw[<-] (x) to node [sloped, above] {$\pi_\cI$} (s);
  \draw[->, dashed, right] (s) to node {$\delta$} (xy);
  \draw[->] (xy) to node [below] {$\pi_\cI$} (x);
\end{tikzpicture}
\end{wrapfigure}
%Our first goal is to show that DEL can be naturally used to reason about task solvability in asynchronous 
%fault-tolerant distributed systems. 
to  represent the task itself.
A  Kripke  model  $\cI$  represents the possible initial states of the  system.
%, as defined by the task to be solved.
The task is specified by an action model $\cT$, which describes the output values
that the agents should be able to produce, as well as preconditions specifying
which inputs are allowed to produce  which outputs.
The product update of the input model $\cI$ with $\cT$ yields an epistemic model $\cI[\cT]$
representing the knowledge that the agents should acquire to solve the task. 
Once the task is specified, given  an action
model $\cA$ that represents some distributed protocol,
the product update of $\cI$ with $\cA$ yields a  Kripke model $\cI[\cA]$
that models how agents perceive the  world after the protocol has been executed.
The protocol $\cA$ \emph{solves} the task if there exists a
morphism $\delta$ that makes the above diagram commute (Definition~\ref{thm:Kripketasksolv2}).
This intuitively happens when  there is sufficient knowledge in $\cI[\cA]$ to solve
the task.
%We formalize this idea by a certain \emph{Kripke morphism}, intuitively saying that
%there is sufficient knowledge in $P$ if and only if there exists such a morphism that  makes a certain 
%diagram commute (Definition \ref{thm:Kripketasksolv2}).

%We stress that we prove that in a precise categorical  sense, simplicial complex models and Kripke models
%are equivalent. We also prove that as action models are constructed through 
%categorical products, we can define use simplicial
%complexes as models
%of DEL in the same way as Kripke models and action models are used for giving
%semantics of DEL~: simplicial complexes are sound and complete models of DEL
%as well. 
%%by moving to the dual version of DEL in terms of simplicial complexes.
%Doing so, we make explicit algebraic topological properties which determine
%knowledge gain by the agents, and transport all the methods that have been
%used successfully in e.g. the algebraic topological approach to fault-tolerant
%distributed protocols \cite{HerlihyKR:2013} to the realm of DEL. % \SR{polish this paragraph}

Beyond the applications that we provide in this paper,
our main goal is to construct a general framework
that connects epistemic logic and distributed computability.
In one direction, uncovering the higher-dimensional topological structure hidden
in Kripke models allows us to transport  methods that have been
used successfully in  the algebraic topological approach to fault-tolerant
distributed computability~\cite{HerlihyKR:2013} to the realm of DEL. 
 In particular, the knowledge gained by applying an action model
is intimately related to how well it preserves the topology of the initial model.
The benefit in the other direction is in providing a formal epistemic logic
semantics to distributed task computability. This allows one to understand better
the abstract topological arguments in terms
of how much knowledge is necessary to solve a task.

%In one direction,  the benefit is in providing for the first time a formal semantics to distributed task computability,
% and in the other direction, the connection  brings in the topological invariants known in distributed computing,
%to DEL, and in particular shows that knowledge gained after an epistemic action model
%is intimately related to  higher dimensional topological properties. 

 We concentrate on the specific setting of asynchronous wait-free shared read/write memory. 
However, there are known equivalences between task solvability in our model and other shared 
memory and message passing models, and this model can be used as a basis
to study task solvability in other more complex models, e.g.
where the number of processes that can crash is bounded
or even where  Byzantine failures are possible~\cite{HerlihyKR:2013}.
%~\cite{ABD1995}.
Nevertheless, this is far from telling the whole story. 
% Our formalisation represents only the case of what the agents should learn to solve a task. 
In the Conclusions section we discuss  many interesting avenues that remain to be explored. 
And additional technical details appear in the companion Technical Reports~\cite{ericSergioDEL1-2017,ericSergioDEL2-2017}.

\paragraph*{Related work. }
Work on knowledge and distributed systems is of course one of the inspirations
of the present work~\cite{FHMVbook},  especially where connectivity~\cite{CastanedaGM:2016,CastanedaGM:2014}
is used.
But the authors know of no previous work  using DEL~\cite{sep-dynamic-epistemic,DEL:2007} to study such systems,
and neither on directly  connecting the combinatorial
topological setting  of \cite{HerlihyKR:2013} with Kripke models. 
In \cite{Hirai}, the author proposes a variant of (non dynamic) epistemic
logic 
for a restricted form of wait-free task specification that 
cannot account for important tasks such
as consensus. 
Similar to~\cite{MosesPreC15}, we show that even though a problem
may not explicitly mention the agents' knowledge, it can  in fact be restated as  knowledge gain requirements.
Nevertheless,  we exploit the ``runs and systems" framework in an orthogonal way, and the knowledge
requirements we obtain are about inputs;
common knowledge in the case of consensus, but other forms of nested knowledge for other tasks.
In contrast,  the  \emph{knowledge of precondition principle} of~\cite{MosesPreC15}
implies that  common knowledge is a necessary condition
for performing simultaneous actions.
%to formalize knowledge and knowledge changes in 
%Action model logic is a dynamic epistemic logic~\cite{DEL:2007}
%formalizing knowledge and knowledge changes in
%fault tolerant distributed settings. 
%Epistemic actions describe knowledge changes by defining different ways of communication and interaction. 
% We use \emph{action model logic}~\cite{BaltagM2004,baltag&:98}. Various examples of epistemic actions have been considered, especially 
%\emph{public announcement logic},  a well-studied example of DEL, with many applications in dynamic logics, knowledge representation and other
%formal methods areas.
%Such change is caused by communication
%Epistemic logic was extended to model public announcements~\cite{plaza:89}.
%Plaza~\cite{plaza:89} first extended , where the same information is transmitted to all agents.
%A variety of approaches include communication that does not necessarily reach all agents, but we have focused  on classical DEL.
Our formulation of carrier maps as products
% (Definition \ref{thm:Kripketasksolv2}) 
has  been partially observed in~\cite{Havlicek2000}.
There are other (categorical) connections between Kripke frames and geometry~\cite{PORTER2004235}.
%We have focused in this paper on the classical semantics of multi-modal S5 logics and its intimate
%connections with topological methods in distributed computing. 
%ormulti-agent epistemic logics using Kripke models. 

DEL is often thought of as  inherently being capable of modeling
only  agents that are synchronous, but as discussed in~\cite{Degremont2011},
this is not the case. More recently,~\cite{knightMS2017} proposes 
a variant of public announcement logic for asynchronous systems that
 introduces two different modal operators for sending and receiving messages.
As we show here, DEL can naturally model the knowledge
in an asynchronous %zfailure-prone
distributed system, at least as far as it is concerned with task solvability.
Further work is needed to study more in depth the knowledge that is
represented in this way.

\section{A simplicial model for epistemic logic}
\label{sec:categories}
\comment{Kripke frames, and Kripke models can be organized as categories. The interest is that the semantics of distributed systems can be expressed using categorical operators on
Kripke frames. 
}

We describe here  the new kind of model for epistemic logic, based on 
chromatic simplicial complexes. 
\comment{% not room for in conference version to repeat this:
We will show that these simplicial models are very closely 
related to the usual Kripke models which are widely used to study the semantics 
of epistemic logic formulas: there is an equivalence of categories between the
two structures.
The geometric nature of simplicial complexes allows us to consider higher-dimensional topological
properties of our models, and investigate their meaning in terms of knowledge.
The idea of using simplicial complexes comes from distributed computability~\cite{HerlihyKR:2013,kozlov:2007};
}
The link with DEL and distributed computing will be developed in the next sections.

\paragraph{Syntax.}
Let $\I{AP}$ be a countable set of propositional variables and $A$ a finite set of agents.
The language $\mathcal{L}_K$ is generated by the following BNF
grammar:
\[
\varphi ::= p \mid \neg\varphi \mid (\varphi \land \varphi) \mid
K_a\varphi \qquad p \in \I{AP},\ a \in A
\]
In the following, we work with $n+1$ agents, and write $A = \{ a_0, \ldots, a_n \}$.

\paragraph{Simplicial complexes and Kripke frames.}

%\begin{definition} [Simplicial complex \cite{kozlov:2007}]
Given a set $V$, a \emph{simplicial} complex $C$ is a family of non-empty finite subsets of $V$ such that
for all $X \in C$, $Y \subseteq X$ implies $Y \in C$. We say $Y$ is a \emph{face} of $X$.
%\end{definition}
Elements of $V$ (identified with singletons) are called \emph{vertices}.
Elements of $C$ are \emph{simplexes}, and those which are maximal w.r.t.\ inclusion are \emph{facets}.
The set of vertices of $C$ is noted $\cV(C)$, and the set of facets $\cF(C)$.
 The \emph{dimension} of a simplex $X \in C$ is  $|X|-1$. 
 A simplicial complex $C$ is \emph{pure} if all its facets are of the same dimension, $n$. In this case,
 we say $C$ is of dimension $n$.
Given the set $A$ of agents (that we will represent as colors), a \emph{chromatic simplicial complex} $\la C, \chi \ra$ consists of a simplicial complex $C$ and a coloring map $\chi : \cV(C) \to A$, such that for all $X \in C$, all the vertices of $X$ have distinct colors.

%\begin{definition} [Simplicial maps \cite{kozlov:2007}] 
Let $C$ and $D$ be two simplicial complexes. A \emph{simplicial map} 
$f: C \rightarrow D$ 
maps the vertices of $C$ to vertices of $D$, such that if $X$ is a simplex of $C$, $f(X)$ %(the image set
%on the subset of vertices $X$) 
is a simplex of $D$. 
%\end{definition}
A \emph{chromatic simplicial map} between two chromatic simplicial complexes is a simplicial map that preserves colors.
Let ${\cal S}_A$ be the category of pure chromatic simplicial complexes on $A$, 
with chromatic simplicial maps for morphisms. 

\medskip

A \emph{Kripke frame} $M=\la S, \sim \ra$ over a set $A$ of agents consists of  a set
of states $S$ and a family of equivalence relations on $S$, written $\sim_a$ for every $a \in A$.
Two states $u, v \in S$ such that $u \sim_a v$ are said to be \emph{indistinguishable} by $a$.
A Kripke frame is \emph{proper} if any two states can be distinguished by at least one agent.
Let $M=\la S, \sim \ra$ and $N=\la T,\sim' \ra$ be two Kripke frames.
A \emph{morphism} from  $M$ to $N$ is a function $f$ from $S$ to $T$ such that for all $u, v \in S$, for all $a \in A$, 
$u \sim_a v$ implies $f(u) \sim'_a f(v)$. 
We write ${\cK}_A$ for the category of proper Kripke frames, with morphisms of Kripke frames as arrows.

\medskip

The following theorem states that we can canonically associate a Kripke frame with a pure chromatic simplicial complex, and vice versa. %\JL{find better names than F and G?}
In fact, this correspondence extends to morphisms, and thus we have an equivalence of categories, meaning that the two structures contain the same information.

\begin{theorem}
\label{thm:equiv}
${\cal S}_A$ and ${\cal K}_A$ are equivalent categories.
\end{theorem}
\begin{proof}[Proof.]
We construct functors $F: {\cal S}_A \rightarrow {\cal K}_A$ and $G: {\cal K}_A \rightarrow {\cal S}_A$ as
follows. 

Let $C$ be a pure chromatic simplicial complex on the set of agents $A$. Its associated Kripke frame is $F(C)=\la S, \sim \ra$, where $S$ is the set of facets of $C$, and the equivalence relation $\sim_a$, for each $a \in A$, is generated by
the relations $X \sim_a Y$ (for $X$ and $Y$ facets of $C$) if $a \in \chi(X \cap Y)$.

For a morphism $f: C \rightarrow D$ in ${\cal S}_A$, we define $F(f) : F(C) \to F(D)$ that takes a facet $X$ of $C$ to its image $f(X)$, which is a facet of $D$ since $f$ is a chromatic map. Assume $X$ and $Y$ are facets of $C$ such 
that $X \sim_a Y$ in $F(C)$, that is, $a \in \chi(X \cap Y)$.
So there is a vertex $v \in \cV(C)$ such that $v \in X \cap Y$ and $\chi(v) = a$.
Then $f(v) \in f(X) \cap f(Y)$ and $\chi(f(v)) = a$, so $a \in \chi(f(X) \cap f(Y))$.
Therefore, $f(X) \sim_a f(Y)$, and $F(f)$ is a morphism of Kripke frames.

Conversely, consider a Kripke frame $M=\la S,\sim \ra$ on the set of agents $A = \{a_0, \ldots, a_n\}$.
Intuitively, what we want to do is take one $n$-simplex $\{v^s_0,\ldots,v^s_n\}$ for each $s \in S$, and glue them together according to the indistiguishability relation.
Formally, let $V = \{ v^s_i \mid s \in S, 0 \leq i \leq n \}$, and equip it with the equivalence relation $R$ defined by $v^s_i \mathrel{R} v^{s'}_i$ if and only if $s \sim_{a_i} s'$.
Then define $G(M)$ whose vertices are the equivalence classes $[v^s_i] \in V/R$, and whose simplexes are of the form $\{[v^s_0],\ldots,[v^s_n]\}$ for $s \in S$, as well as their sub-simplexes. The coloring map is given by $\chi([v^s_i])=a_i$.
It is a well-defined chromatic simplicial complex since all elements of an
equivalence class of $R$ have the same color.
%It is pure since, as
%we equate only vertices with the same color, and colors are in bijection with
%extremal points in all simplexes, we cannot equate a simplex with a lower
%dimensional simplex.
The facets are exactly the $\{[v^s_0],\ldots,[v^s_n]\}$ for $s \in S$, since the Kripke frame $M$ is proper, we cannot equate two facets together.

Now let $f : M \to N$ be a morphism in ${\cal K}_A$. We define $G(f) : G(M) \to G(N)$ that maps a vertex $[v^s_i]$ of $G(M)$ to the vertex $[v^{f(s)}_i]$ of $G(N)$.
This map is well-defined (i.e., the image of a vertex does not depend on the chosen representative) because $f$ is a morphism of Kripke frames, and thus it preserves the indistinguishability relations.
It is easily checked that this is moreover a simplicial map.

Consider now a Kripke frame $M=\la S,\sim \ra$ in ${\cal K}_A$ with agent set $A$. 
$FG(M)$ is the Kripke frame $N=\la T,\sim' \ra$ such that $T$ is the set of facets of
$G(M)$. But we have seen above that the facets of $G(M)$ are of the form $\{[v^s_0],\ldots,[v^s_n]\}$ (where $s \in S$), therefore, $T$ is in bijection with
$S$. Finally, in $FG(M)$, $X \sim'_a Y$ if and only if $a \in \chi(X \cap Y)$, where
$\chi$ is the coloring, in $G(M)$, of $X$ and $Y$ which are facets in $G(M)$. But facets
in $G(M)$ are just in direct bijection with the worlds of $M$, i.e. $X=\{[v^s_0],\ldots,
[v^s_n]\}$ and $Y=\{[v^t_0],\ldots,[v^t_n]\}$ where $s, t \in M$. Note that $\chi([v^s_i])=a_i$
and $\chi([v^t_i])=a_i$ so $a \in \chi(X\cap Y)$ means that $a=a_i$ for some $i$
and $v^s_i \mathrel{R} v^t_i$. This can only be the case, by definition of $G(M)$ if $s \sim_{a_i} t$. 
This proves that $FG(M)$ and $M$ are isomorphic Kripke frames. 

Consider now a pure chromatic simplicial complex $C \in {\cal S}_A$. 
It is easily seen that $GF(C)$ is isomorphic, as a pure chromatic simplicial complex, 
to $C$, hence ${\cal S}_A$ and ${\cal K}_A$ are equivalent categories. 
\end{proof}

\begin{example}
The picture below shows a Kripke frame (left) and its associated chromatic
simplicial complex (right). The three agents, named $b, g, w$, are represented as
colors black, grey and white on the vertices of the simplicial complex. The three worlds of the
Kripke frame correspond to the three triangles (i.e., $2$-dimensional simplexes)
of the simplicial complex. The two worlds indistinguishable by agent
$b$, are glued along their black vertex; the two worlds indistinguishable by agents $g$ and $w$ are glued along the grey-and-white edge.

\begin{center}
\begin{tikzpicture}[auto,dot/.style={draw,circle,fill=black,inner sep=0pt,minimum size=3pt},cloudgrey/.style={draw=black,thick,circle,fill={rgb:black,1;white,2},inner sep=0pt,minimum size=8pt},cloud/.style={draw=black,thick,circle,fill=white,inner sep=0pt,minimum size=8pt}, cloudblack/.style={draw=black,thick,circle,fill=black,inner sep=0pt,minimum size=8pt}]
\node[dot] (p) at (-1,0) {};
\node[dot] (q) at (0,0) {};
\node[dot] (r) at (1,0) {};
\draw (p) -- node[above] {$g,w$} (q);
\draw (q) -- node[above] {$b$} (r);
 
\path[->, bend right, >=stealth] (3.5,0.2)  edge node[above] {$F$} (1.5,0.2);
\path[->, bend right, >=stealth] (1.5,-0.2)  edge node[below] {$G$} (3.5,-0.2);

\draw[thick, draw=black, fill=blue, fill opacity=0.15]
  (4,0) -- (5,-0.577) -- (5,0.577) -- cycle;
\draw[thick, draw=black, fill=blue, fill opacity=0.15]
  (5,-0.577) -- (5,0.577) -- (6,0) -- cycle;
\draw[thick, draw=black, fill=blue, fill opacity=0.15]
  (6,0) -- (7,-0.577) -- (7,0.577) -- cycle;
\node[cloudblack] (b1) at (4,0) {};
\node[cloudgrey] (g1) at (5,-0.577) {};
\node[cloud] (w1) at (5,0.577) {};
\node[cloudblack] (b2) at (6,0) {};
\node[cloudgrey] (g2) at (7,-0.577) {};
\node[cloud] (w2) at (7,0.577) {};
\end{tikzpicture}
\end{center}
\end{example}

%Thus, according to Theorem~\ref{thm:equiv}, chromatic simplicial complexes and Kripke frames contain the same information.
We now decorate our simplicial complexes with atomic propositions in order to 
get a notion of simplicial model.

\paragraph{Simplicial models and Kripke models.}
For technical reasons, we restrict to models where all the atomic propositions
are saying something about some local value held by one particular agent. All the examples 
that we are interested in will fit in that framework. Let $\Values$ be some 
countable set of values, and $\I{AP} = \{ p_{a,x} \mid a \in A, x \in \Values \}$ 
be the set of \emph{atomic propositions}. Intuitively, $p_{a,x}$ is true if agent 
$a$ holds the value $x$.
%In all our examples, an agent will hold exactly one value, but we do not enforce that.
We write $\I{AP}_a$ for the atomic propositions concerning agent $a$.

%A \emph{literal} over $\I{AP}$ is an element of $\Lit(\I{AP}) = \{p, \neg p \mid p \in \I{AP}\}$.
%A set $X$ of literals is \emph{consistent} if it does not contain both $p$ and $\neg p$ for any $p$; and $X$ is \emph{$\I{AP}$-maximal} if
%for all $p$, either $p \in X$ or $\neg p \in X$.

A \emph{simplicial model} $M = \la C, \chi, \ell \ra$ consists of a pure chromatic simplicial 
complex $\la C, \chi \ra$ of dimension $n$, and a labeling $\ell : \cV(C) \to \mathscr{P}(\I{AP})$
that associates with each vertex $v \in \cV(C)$ a set of atomic propositions concerning agent $\chi(v)$, i.e., such that $\ell(v) \subseteq \I{AP}_{\chi(v)}$.
Given a facet $X = \{v_0, \ldots, v_n\} \in C$, we write $\ell(X) = \bigcup_{i=0}^n \ell(v_i)$.
A \emph{morphism} of simplicial models $f : M \to M'$ is a chromatic simplicial map that preserves the labeling: $\ell'(f(v)) = \ell(v)$ (and $\chi$).
We denote by ${\cal SM}_{A,\I{AP}}$ the category of simplicial models over the set of 
agents $A$ and atomic propositions $\I{AP}$.

\medskip

A \emph{Kripke model}  $M = {\la S,\sim,L\ra}$ consists of a Kripke frame
 $\la S, \sim \ra$  and a function $L : S \to \mathscr{P}(\I{AP})$.
Intuitively, $L(s)$ is the set of atomic propositions that are true in the state $s$.
%
% The \emph{knowledge} $K_a$ of an agent $a$ with respect to
% a state $s$ is the set of formulas which are true in all states that are
% indistinguishable from $s$ by $a$.
%
A Kripke model is \emph{proper} if the underlying Kripke frame is proper.
%
%As in the case of simplicial models, we now fix the set of atomic propositions
%$\I{AP} = \{ p_{a,x} \mid a \in A, x \in \Values \}$, and for a given agent $a$, 
%$\I{AP}_a = \{ p_{a,x} \mid x \in \Values \}$. 
A Kripke model is \emph{local} if for every agent $a \in A$, $s \sim_a s'$ implies
$L(s) \cap \I{AP}_a = L(s') \cap \I{AP}_a$, i.e., an agent always 
knows its own values.

Let $M=\la S, \sim,L \ra$ and $M'=\la S',\sim',L' \ra$ be two Kripke
models on the same set $\I{AP}$. A \emph{morphism of Kripke models} $f : M \to M'$
is a morphism of the underlying Kripke frames such that 
$L'(f(s)) =  L(s)$ for every state $s$ in $S$. 
We write $\cal KM_{A, \I{AP}}$ for the category of local proper Kripke models. %with such morphisms.

\medskip
We can now extend the two maps $F$ and $G$ of Theorem~\ref{thm:equiv} to an equivalence between simplicial models and Kripke models.
%On the underlying Kripke frame and simplicial complex, they act the same.
%Given a simplicial model ${M = \la C, \chi, \ell \ra}$, we associate the Kripke model $F(M) = \la \cF(C), 
%\sim, L \ra$ where the labeling~$L$ of a facet $X \in \cF(C)$ is given by $L(X) = \bigcup_{v \in X} \ell(v)$.
%%
%Conversely, given a Kripke model $M = \la S, \sim, L \ra$, the vertex $v^s_i$ 
%(colored by $a_i$) of $G(M)$ is labeled by $\ell(v^s_i) = L(s) \cap \I{AP}_{a_i}$.

\begin{theorem}
\label{thm:equiv-model}
${\cal SM}_{A,\I{AP}}$ and ${\cal KM}_{A,\I{AP}}$ are equivalent categories.
\end{theorem}
\begin{proof}[Proof.]
%The two maps $F$ and $G$ defined above can be extended to functors that form an equivalence of categories between ${\cal SM}_{A,\I{AP}}$ and ${\cal KM}_{A,\I{AP}}$ (see Appendix~\ref{app:kripkeModels}).
We describe the functors $F : {\cal SM} \to {\cal KM}$ and $G : {\cal KM} \to {\cal SM}$. On the underlying Kripke frame and simplicial complex, they act the same as in the proof of Theorem~\ref{thm:equiv}.

Given a simplicial model $M = \la C, \chi, \ell \ra$, we associate the Kripke model $F(M) = \la \cF(C), 
\sim, L \ra$ where the labeling $L$ of a facet $X \in \cF(C)$ is given by $L(X) = \bigcup_{v \in X} \ell(v)$. This Kripke model is local since $X \sim_a Y$ means that $X$ and $Y$ share an $a$-colored vertex $v$, so $L(X) \cap \I{AP}_a = L(Y) \cap \I{AP}_a = \ell(v)$.

Conversely, given a Kripke model $M = \la S, \sim, L \ra$, the underlying 
simplicial complex of $G(M)$ is obtained by gluing together $n$-simplexes of the 
form $\{v^s_0,\ldots,v^s_n\}$, with $s \in S$. We label the vertex $v^s_i$ 
(colored by $a_i$) by $\ell(v^s_i) = L(s) \cap \I{AP}_{a_i}$. This is well 
defined because two vertices $v^s_i$ and $v^{s'}_i$ are identified whenever $s \sim_{a_i} s'$,
so $L(s) \cap \I{AP}_{a_i} = L(s') \cap \I{AP}_{a_i}$ since $M$ is local.

The action of $F$ and $G$ on morphisms is the same as in Theorem~\ref{thm:equiv}.
It is easy to check that the additional properties of morphisms between models 
are verified. Checking that $FG(M) \simeq M$ and $GF(M) \simeq M$ also works the 
same as in the previous theorem.
\end{proof}

\begin{example}\label{ex:binInputs}
The figure below shows the so-called binary input complex and its associated Kripke model, for 2 and 3 agents.
Each agent gets a binary value $0$ or $1$, but doesn't know which value has been received by the other agents. So, every possible combination of $0$'s and $1$'s is a possible world.

In the Kripke model, the agents are called $b, g, w$, and the labeling $L$ of the
possible worlds is represented as a sequence of values, e.g., $101$,
representing the values chosen by the agents $b, g, w$ (in that order).
In the 3-agents case, the labels of the dotted edges have been omitted to avoid
overloading the picture, as well as all the edges labeled by only one agent.

In the simplicial model, agents are represented as colors (black, grey, and white). The labeling $\ell$ is represented as a single value in a vertex, e.g.,
``$1$'' in a grey vertex means that agent $g$ has chosen value~$1$.
The possible worlds correspond to edges in the 2-agents case, and triangles in the 3-agents case.
%In the Figures below %~\ref{fig-spheres1} and \ref{fig-spheres2}
%the binary chromatic complexes including
% all combinations of 0's and 1's are illustrated, for 2 and 3 processes. 
%%; the simplex of all 1's
%%should be removed to obtain the complex of Example~\ref{ex:RWsystems}. 
%Each vertex of a simplex has a color
%that represents one of the processes, and a number that is its input value. States of Kripke frames are
%associated to facets of the corresponding 
%simplicial complex. %(under Theorem \ref{thm:equiv}). 
%When 
%two states are equivalent under $\sim_i$, we identify the vertices with color $i$ in the two corresponding facets. 
%\begin{center}
%\begin{figure}[h]

\medskip
\noindent
\begin{minipage}{0.4\textwidth}
%\begin{tabular}{ccc}
\begin{center}
\begin{tikzpicture}[auto,line/.style={draw,thick,-latex',shorten >=2pt},cloudgrey/.style={draw=black,thick,circle,fill={rgb:black,1;white,3},minimum height=1em},cloud/.style={draw=black,thick,circle,fill=white,minimum height=1em}]
\matrix[column sep=2mm,row sep=2mm]
{
& \node (alpha'') {$01$}; & & \node [cloudgrey] (u) {$0$}; & \node[circle] { }; & \node [cloud] (v) {$1$}; \\
\node (beta'') {$00$}; & & \node (gamma'') {$11$}; & & & \\
& \node (delta'') {$10$}; & & \node [cloud] (w) {$0$}; & & \node [cloudgrey] (z) {$1$}; \\
};
\draw (alpha'') -- node[above left,xshift=2pt,yshift=-2pt] {$\scriptstyle g$} (beta'');
\draw (beta'') -- node[below left,xshift=2pt,yshift=2pt] {$\scriptstyle w$} (delta'');
\draw (gamma'') -- node[xshift=-2pt,yshift=2pt] {$\scriptstyle g$} (delta'');
\draw (alpha'') -- node[xshift=-2pt,yshift=-2pt] {$\scriptstyle w$} (gamma'');
\draw[semithick] (u) -- (v) -- (z) -- (w) -- (u);
\end{tikzpicture}
%\caption{Kripke frame for 2 processes with binary inputs and its dual complex.}
%\label{fig-spheres1}
\end{center}
\end{minipage}
\hfill
\begin{minipage}{0.5\textwidth}
\begin{center}
\begin{tikzpicture}[auto,rotate=45,font=\scriptsize]
\node (a) at (0,0) {$111$};
\node (b) at (1.5,0) {$011$};
\node (c) at (0,1.5) {$101$};
\node (d) at (1.5,1.5) {$001$};
\node (e) at (0.6,0.6) {$110$};
\node (f) at (2.1,0.6) {$010$};
\node (g) at (0.6,2.1) {$100$};
\node (h) at (2.1,2.1) {$000$};
\path[inner sep = 0pt]
      (a) edge node[below right] {$gw$} (b)
          edge node {$bw$} (c)
          edge[dotted] (e) %node[right,xshift=-2pt] {$gb$} (e)
      (d) edge node[above left,pos=0.4,xshift=2pt] {$gw$} (c)
          edge node[pos=0.35] {$bw$} (b)
          edge node[right,xshift=1pt,pos=0.3] {$bg$} (h)
      (f) edge node[above right] {$bw$} (h)
          edge node[xshift=1pt] {$bg$} (b)
          edge[dotted] (e) % node {$wg$} (e)
      (g) edge node {$gw$} (h)
          edge node[left,xshift=-1pt] {$bg$} (c)
          edge[dotted] (e); % node {$wb$} (e);
\end{tikzpicture}
\raisebox{1em}{\includegraphics[scale=0.25]{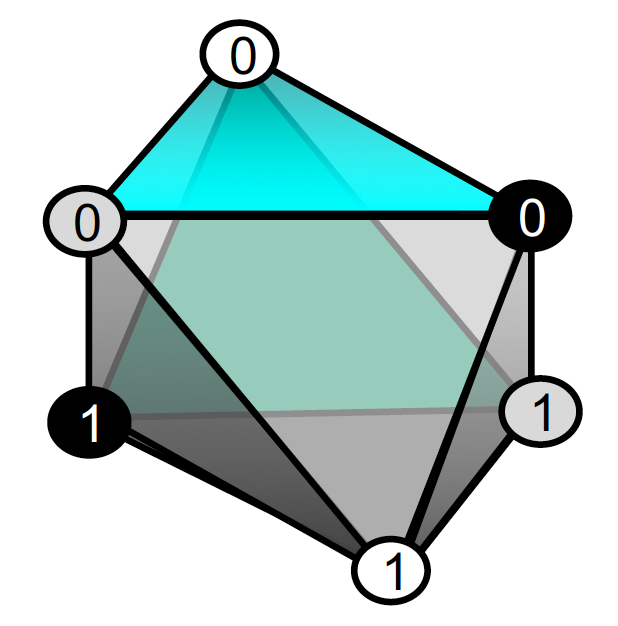}}
%\caption{Kripke frame for 3 processes with binary inputs (not all edges depicted) and its dual complex.}
%\label{fig-spheres2}
\end{center}
\end{minipage}
%\end{figure}

\medskip
\noindent
It is well known in the context of distributed computing~\cite{HerlihyKR:2013} that the binary input simplicial complex for $n+1$ agents is a $n$-dimensional sphere.
\end{example}

\begin{example}\label{ex:torus}
Consider the following situation. There are three agents black, grey and white,
and a deck of four cards, $\{0, 1, 2, 3\}$. One card is given to each agent, 
and the last card is kept hidden. Each agent knows its own card, but
not the other agents' cards.
The simplicial model corresponding to that situation is depicted
below on the left. The color of vertices indicate the corresponding agent, and the labeling is its card.
In the planar drawing, vertices that appear several times with the same color and
value should be identified. The arrows $A$ and $B$ indicate how the edges should
be glued together. What we obtain is a triangulated torus.

If the deck of cards is $\{0, 1, 2\}$,
we get the figure on the right, where the two white vertices (with card $0$)
should be identified, as well as the two black ones.
\begin{center}
\includegraphics[scale=0.25]{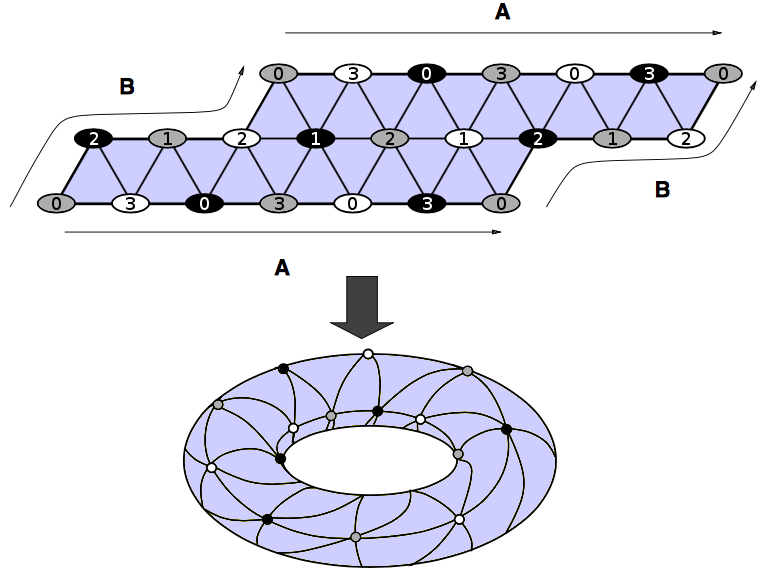}
\qquad
\includegraphics[scale=0.20]{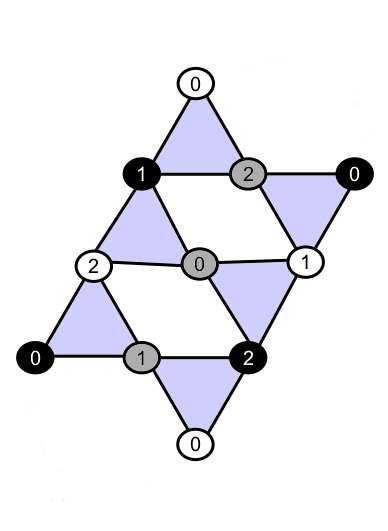}
%\caption{The simplicial model for 3 players and 4 cards}
\label{fig:torus}
\end{center}
\end{example}

\medskip
Thus, Theorem~\ref{thm:equiv-model} says that simplicial models are closely related to Kripke models. Keeping that 
translation in mind, we can reformulate the usual semantics of formulas in Kripke 
models, in terms of simplicial models.

\begin{definition}\label{semantFormulas}
We define the truth of a formula $\varphi$ in some epistemic state $(M,X)$ with $M=\la C, \chi, \ell \ra$ a simplicial model, $X \in \cF(C)$ a facet of $C$ and
$\varphi \in \Lcal_K(A,\I{AP})$.
The satisfaction relation, determining when a formula is true in an
epistemic state, is defined as:

\begin{tabular}{lrl}
$M,X \models p$ & iff & $p \in \ell(X)$\\
$M,X \models \neg \varphi$ & iff & $M,X \not\models \varphi$\\
$M,X \models \varphi \wedge \psi$ & iff & $M,X \models \varphi
\mbox{ and } M,X \models \psi$\\
$M,X \models K_a \varphi$ & iff & $\mbox{for all } Y \in \cF(C), a \in \chi(X \cap Y) \mbox{ implies } M,Y \models \varphi$\\
\end{tabular}
\end{definition}

We can show that this definition of truth agrees with the usual one (which we write $\models_\cK$ to avoid confusion) on the corresponding Kripke model.

\begin{proposition} \label{prop:truth}
Given a simplicial model $M$ and a facet $X$, $M,X \models \varphi$ iff $F(M), X \models_\cK \varphi$. Conversely, given a local proper Kripke model $N$ and state $s$, $N,s \models_\cK \varphi$ iff $G(N),G(s) \models \varphi$, where $G(s)$ is the facet $\{v_0^s, \ldots, v_n^s\}$ of $G(N)$.
\end{proposition}
\begin{proof}
This is straightforward by induction on the formula $\varphi$.
\end{proof}

It is well-known that the axiom system $\mathbf{S5_n}$ is sound and complete with respect to the class of Kripke models~\cite{DEL:2007}. 
Since we restrict here to local Kripke models, we need to add the following axiom (or axiom schema, if $\Values$ is infinite), saying that every agent knows which values it holds:

\[\mathbf{Loc}\; = \bigwedge_{a\in A, x\in\Values} K_a(p_{a,x}) \lor K_a(\neg p_{a,x})\]

\begin{corollary}
The axiom system $\mathbf{S5_n}+\mathbf{Loc}$ is sound and complete w.r.t.\ the class of simplicial models.
\end{corollary}
\begin{proof}
Adapting the proof of~\cite{DEL:2007} for $\mathbf{S5_n}$, it can be shown that
$\mathbf{S5_n}+\mathbf{Loc}$ is sound and complete w.r.t.\ the class of local
proper Kripke models. %(see Appendix~\ref{app:sound-complete}). 
Then, we transpose it to simplicial models using
Proposition~\ref{prop:truth}.

Indeed, suppose a formula $\varphi$ is true for every local proper Kripke model and any state.
Then given a simplicial model and facet $(M,X)$, since by assumption $F(M), X \models_\cK \varphi$, we also have $M,X \models \varphi$ by Proposition~\ref{prop:truth}.
So $\varphi$ is true in every simplicial model.
Similarly, the converse also holds.
\end{proof}

The following Theorem shows that morphisms of simplicial models cannot ``gain knowledge about the world''. This will be useful in Section~\ref{sec:modeKripkeF} when we formulate the solvability of a task as the existence of some morphism.

\begin{theorem}[knowledge gain] %\JL{This theorem is false, we must moreover ask that $\varphi$ doesn't contain knowledge operators}
\label{thm:lose-knowledge}
Consider simplicial models $M=\la C, \chi, \ell \ra$ and
$M' = \la C', \chi', \ell' \ra$, and a  morphism $f : M \to M'$.
Let $X \in \cF(C)$ be a facet of $M$, $a$ an agent, and $\varphi$ a formula 
which does not contain negations except, possibly, in front of atomic propositions.
Then, $M',f(X) \models \varphi$ implies $M,X \models \varphi$.
\end{theorem}
\begin{proof}
%This is done by induction on $\varphi$.
%See Appendix \ref{proof:thm-lose-knwledge}.
We proceed by induction on $\varphi$.
%\begin{itemize}
%\item 
First, for $p$ an atomic proposition, since morphisms preserves the valuation $\ell$, we
have $M',f(Y) \models p$ iff $M,Y \models p$. Thus the theorem is true for
(possibly negated) atomic propositions.
%\item 
The case of the conjunction follows trivially from the induction hypothesis.
%\item 

Suppose now that $M',f(X) \models K_a \varphi$.
In order to show $M,X \models K_a \varphi$, assume that $a \in \chi(X \cap Y)$
for some facet $Y$, and let us prove $M,Y \models \varphi$.
Let $v$ be the $a$-colored vertex in $X \cap Y$.
Then $f(v) \in f(X) \cap f(Y)$ and $\chi(f(v)) = a$.
So $a \in \chi(f(X) \cap f(Y))$ and thus $M', f(Y) \models \varphi$.
By induction hypothesis, we obtain $M,Y \models \varphi$.
Finally, suppose that $M',f(X) \models C_B \varphi$.
We want to show that $M,X \models C_B \varphi$, i.e., for every $Y$ reachable from $X$ following a sequence of simplexes sharing a $B$-colored vertex, $M,Y \models \varphi$. By the same reasoning as in the $K_a$ case, $f(Y)$ is $B$-reachable from $f(X)$, so $M',f(Y) \models \varphi$, and thus $M,Y \models \varphi$.
%\end{itemize}
\end{proof}

The restriction on $\varphi$ forbids formulas saying something about what an agent
does not know. Indeed, one can ``gain'' the knowledge that some agent does not
know something; but this is not relevant information for solving the tasks that
we have in mind.
In fact, Theorem~\ref{thm:lose-knowledge} is still true if the formula $\varphi$
contains other knowledge operators such as group and common knowledge. 
For~$B$ a subgroup of agents,
group knowledge is defined as 
$E_B \varphi = \bigwedge_{b\in B} K_b \varphi$ and 
common knowledge for group~$B$ is, semantically, the least solution of the equation
$C_B \varphi = \varphi \wedge E_B(C_B \varphi)$.

%%%%%%%%%%%%%%%%%%%%%%%%%%%%%%%%%%%%%%%%%%%%%%%%%%%

% !TEX root =  ./main.tex

\section{DEL via simplicial complexes}
\label{subsec:model}
%Additional details are in the Appendix.
We describe here our adaptation of Dynamic Epistemic Logic (DEL) to simplicial models, and
an  action model that is fundamental in distributed computing.

\subsection{DEL basic notions}
\label{sec:DELbasics}
DEL is the study of modal logics of model change~\cite{sep-dynamic-epistemic,DEL:2007}.
A modal logic studied in DEL is obtained
by using action models~\cite{baltagMS:98}, which are relational structures that can be used to describe a variety of informational actions.
%pag 149 [van Ditmarsch et al.]
%
%\paragraph{Syntax.}
%We extend the syntax of epistemic logic with one more construction:
%\[
%\varphi ::= p \mid \neg\varphi \mid (\varphi \land \varphi) \mid
%K_a\varphi \mid [\alpha]\varphi \qquad p \in \I{AP},\ a \in A
%\]
%Intuitively, $[\alpha]\varphi$ means that $\varphi$ is true after some 
%\emph{action} $\alpha$ has occurred.
An action can be thought of as an 
announcement made by the environment,
which is not necessarily public, in the sense that
not all  agents receive these announcements.
An action model describes all the possible actions that might happen, as well
as how they affect the different agents.
We first recall the usual notion of action model; then describe a dual version, appropriate to represent epistemic 
change in simplicial models.

\paragraph{Dynamic Epistemic Logic.}
%\label{sec:kripke}
An \emph{action model} is a structure $\la T,\sim,\pre \ra$,
where $T$ is a domain of \emph{action points}, such that for
each $a \in A$, $\sim_a$ is an equivalence relation on $T$, and
$\pre : T \to \Lcal_\cK$ is a  function that assigns a
\emph{precondition} $\pre(t)$ to each $t \in T$.
%\end{definition}
%Similarly to epistemic models, we may draw  an action model as a graph  where self loops are omitted and
%each pair of directed edges $(s,s')$ and $(s',s)$ is
%represented by a single undirected edge $\{s,s'\}$.
For an initial Kripke model $M$, the effect of action model $\cA$ is a Kripke model $M[\cA]$.
%\begin{definition}[Product update]
Let $M=\la S, \sim,L \ra$ be a Kripke model and $\cA = \la T,\sim,\pre \ra$
 be an action model. The
 \emph{product update model} is $M[\cA]=\la S[\cA], \sim^{[\cA]},L[\cA] \ra$, where
each world of  $S[\cA]$ is a pair $(s,t)$ with $s\in S,\; t\in T$ such that $\pre(t)$ holds in $s$. 
Then, $(s,t)\sim_a^{[\cA]} (s',t')$ whenever it holds that $s\sim_a s'$
and $t\sim_a t'$.
%(hence its underlying Kripke frame is a sub-frame of the cartesian product of the underlying Kripke frames of $M$ and $A$).
The valuation $L[\cA]$ at a pair $(s,t)$ is just as it was at $s$, i.e., $L[\cA]((s,t)) = L(s)$.
%\end{definition}
%A pair $(s,t)$ is to be included in the worlds of $M[A]$ 
%if and only if $s\in S$ and $M$ satisfies the precondition $pre(A(t))$.
%\footnote{  
%Models in DEL are pointed, but we do not need to state this explicitly.}

\begin{proposition}
\label{prop:actPrLocP}
Let $M$ be a local proper Kripke model and ${\cA=\la T, \sim, \pre\ra}$ a proper action model, then $M[\cA]$ is proper and local.
\end{proposition}
\begin{proof}
$M[\cA]$ is proper: let $(s,t)$ and $(s',t')$ be two distinct states of $M[\cA]$.
Then either $s \neq s'$ or $t \neq t'$, and in both cases, since $M$ and $\cA$ are proper,
at least one agent can distinguish between the two.
Now, $M[\cA]$ is local: suppose $(s,t) \sim^{[\cA]}_a (s',t')$. Then in particular
$s \sim_a s'$ and since $M$ is local, $L(s) \cap \I{AP}_a = L(s') \cap \I{AP}_a$. 
The same goes for $L[\cA]$ since it just copies $L$.
%HOP
%See Appendix~\ref{app:DEL-basics}.
\end{proof}

\paragraph{A simplicial complex version of DEL.}
To work in the category of simplicial models, we consider a simplicial
version of action models.
First, let us define cartesian products.
Given two pure chromatic simplicial complexes $C$ and $T$ of dimension $n$,
the cartesian product $C\times T$ is the following pure chromatic simplicial complex of dimension $n$.
Its vertices are of the form $(u,v)$ with $u \in \cV(C)$ and $v \in \cV(T)$ such that $\chi(u) = \chi(v)$; the color of $(u,v)$ is $\chi((u,v)) = \chi(u) = \chi(v)$.
Its simplexes are of the form $X \times Y = \{(u_0,v_0), \ldots, (u_k,v_k)\}$ where $X = \{u_0, \ldots, u_k\} \in C$, $Y = \{v_0, \ldots, v_k\} \in T$ and $\chi(u_i) = \chi(v_i)$.
%Its facets are simplexes $X\times Y$, $X\in\cF(C),Y\in\cF(T)$, where
%$X\times Y$ consists of $n+1$ pairs of the form $(u,v)$, $u\in X, v\in Y$ such that $\chi(u)=\chi(v)$. 

%Given a relation $\pre\subseteq \cF(C)\times\cF(T)$,
%the induced cartesian product $\pre(C\times T)$ is
%the subcomplex of $C\times T$ of all $c\times a\in \cF(C\times T)$ such that 
%$c\times a\in  \pre$.

A \emph{simplicial action  model},  
$\la T, \chi, \pre\ra$ consists of a pure chromatic simplicial 
complex $\la T, \chi\ra$,
where the facets $\cF(T)$ represent communicative \emph{actions}, and $\pre$ 
assigns to each facet $X\in\cF(T)$ a precondition formula $\pre(X)$ in $\Lcal_\cK$.
%\begin{definition}[Simplicial Product update]
Let  $M = \la C, \chi, \ell \ra$ be a simplicial model, and $\cA=\la T, \chi, \pre\ra$ be a simplicial action model.
 The \emph{product update simplicial model} $M[\cA]= \la C[\cA], \chi[\cA],\ell[\cA] \ra$ is a 
 simplicial model whose underlying simplicial complex is a sub-complex of the
 cartesian product $C \times T$, induced by all the facets of the form $X \times Y$ such that $\pre(Y)$ holds in $X$, i.e., $M,X \models \pre(Y)$. 
% consisting of the pure chromatic simplicial 
%complex $\la C', \chi\ra= \pre(C\times T)$, where abusing notation, $\pre$ is viewed
%as a relation, $\pre \subseteq \cF(C)\times\cF(T)$.
%Thus,
%each facet $f'$ of $C'$ is equal to $f\times t$,  with $\pre(t)$ is true in $f$.
The valuation
$\ell : \cV(C[\cA]) \to \mathscr{P}(\I{AP})$ at a pair $(u,v)$ is just as it was at $u$: $\ell[\cA]((u,v)) = \ell(u)$.
%\SR{Need help stating
% and proving a categorical equivalence of the Kripke action
%model and simplicial action model}
%\EG{As far as I know, there is no notion of (interesting?)
%morphism of Kripke action model although we can cook up one. Maybe we just want to say there is a bijective correspondance between simplicial actions models and Kripke ones? This would be clear by the equivalence between Kripke models and simplicial models. From what I see below, ``equivalent'' means, are in one to one correspondance
%and the interpretation of formulas is the same, but this is clear because
%the equivalence commutes with products, basically.}

%In our logic, an action $\alpha$ is a pair $(\cA,Y)$, where $Y$ is a facet of the
%action model $\cA$. The semantics of formulas can be extended as follows:
%  \[ M,X \models [\cA, Y]\varphi \quad \text{iff} \quad M,X \models \pre(Y) \;\;\text{implies}\;\; M[\cA], X \times Y \models \varphi
%  \]

Recall from Theorem~\ref{thm:equiv-model} the two functors $F$ and $G$ 
that define an equivalence of categories between simplicial models and Kripke models.
We have a similar correspondence between action models and simplicial action
models, which we still write $F$ and $G$. On the underlying Kripke frame and
simplicial complex they are the same as before; and the precondition of an action
point is just copied to the corresponding facet.
The simplicial version of the product update model
agrees with the usual one on Kripke models:

\begin{proposition}
\label{prop:catActMo}
Consider a simplicial model $M$ and simplicial action model $\cA$, and their 
corresponding Kripke model $F(M)$ and action model $F(\cA)$.
Then, the Kripke models $F(M[\cA])$ and $F(M)[F(\cA)]$ are isomorphic.
The same is true for $G$, starting with a Kripke model $M$ and action model $\cA$.
\end{proposition}
\begin{proof}
%See Appendix~\ref{app:DEL-basics}.
The main observation is that both constructions of product update model rely on a
notion of cartesian product (in the category of pure chromatic simplicial complexes
for $M[\cA]$, and in the category of Kripke frames for $F(M)[F(\cA)]$).
These are both cartesian products in the categorical sense, therefore they are
preserved by the functor $F$ because it is part of an equivalence of category.
\end{proof}

\subsection{A basic action model for distributed computing}
\label{sec:basicAcM}

We describe here the \emph{immediate snapshot} action model $\mathcal{IS}$
for one communication
exchange among asynchronous agents. As an action model, it is new and to the best of our knowledge it has not been
studied from the DEL perspective; immediate snapshots operations 
are important  in distributed computing, and many variants of computational models based on them
have been considered, including multi-round communication exchanges, see e.g. \cite{AttiyaR2002,HerlihyKR:2013}
but for the point we want to make about using DEL, 
the main issues can be studied with this very simple action model.

The situation we have in mind is the following.
The $n+1$ agents correspond to $n+1$ concurrent processes.
Initially, each process has some input value, and they communicate (only
once) through a
shared memory array in order to try to learn each other's input value. 
They use the following protocol: each process has a dedicated
memory cell in the array, to which it writes its input value. 
Then, it reads one by one all the cells of the array,  to see which other input values have been written. Based on this information, each process
decides an output value.
The processes are  asynchronous, meaning that an execution consists of an arbitrary interleaving
of the write and read operations of all the processes (one write per process, and $n+1$ reads per process).

We could describe the action model corresponding to this situation,
and present all of our results using it, 
but to illustrate more easily the basic ideas, we define instead 
an action model $\mathcal{IS}$ corresponding to  a subset
of all the executions in the previous situation. And we do so without
loss of generality, because from the task computability perspective,
they are known to be equivalent~\cite{AttiyaR2002}.

The interleavings we consider can be represented by a sequence
of concurrency classes, $c_1,c_2,\ldots, c_m$.
For each concurrency class $c_i$, all the agents in $c_i$ execute their write
operations simultaneously, then all of them execute their read operations
simultaneously, then we move on to the next concurrency class~$c_{i+1}$.
Thus, all the agents in $c_i$ see each other's values, as well as the values of
the agents from the previous concurrency classes.

% we imagine a scheduler selects
%a set of agents $c_1$, and all of them execute their write operations
%simultaneously, and then all of them execute all their read operations
%simultaneously.
%; thus, all the write operations can happen concurrently,
%and then all the read operations, as if taking a snapshot of the
%shared memory.
%We call the combined schedule an
%\emph{immediate snapshot} operation by the agents in $c_1$.
%Then the scheduler
%selects another concurrency class $c_2$
%of agents, disjoint from $c_1$, and it schedules their operations
%in a similar pattern, so that agents in $c_2$ execute
%  concurrently an {immediate snapshot}, and so on, until every
%  agent has executed exactly once an immediate snapshot.

Let us define formally the simplicial action model corresponding to 
such immediate snapshot schedules.
 A \emph{sequential partition} of agents $A$ is a 
sequence $c=c_1,c_2,\ldots,c_m$, of non-empty, disjoint subsets of $A$, 
whose union is equal to~$A$. Each $c_i$ is called
a \emph{concurrency class}. Notice that $1\leq m\leq |A|$,
and when $m=1$ all agents take an immediate snapshot concurrently,
while if $m=|A|$, all take immediate snapshots sequentially.
The agents in a concurrency class $c_j$ learn
the input values of all the agents in earlier concurrency classes $c_i$ for $i \leq j$,
and which agent wrote which value.
% (but they  do not get information about to which concurrency class each
%agent belongs).
%The sequential partition $c$ represents a communication
%action, where the agents in a concurrency class 
%$c_j$ learn of each others input values,
%and also learn of the input values of every agent in every 
%earlier concurrency class,
%$c_i$, with $i<j$. However, they do not learn the input values
%of agents in later concurrency classes.
In particular, agents in $c_m$
learn the inputs of all agents (and there is always at least one such
agent), and if $m=1$, then all agents learn all the values.
Define $\aview_a(c)$ (`A' stands for ``agent'' view) to be the set of agents whose
inputs are seen by $a$ in $c$: if $a \in c_j$, $\aview_a(c) = \bigcup_{i \leq j} c_i$.
Notice that
two executions of the immediate snapshot are indistinguishable by $a$ when the corresponding sequential partitions 
yield the same $\aview$ for $a$, and additionally, the agents in
$\aview$ have the same inputs.
%Thus, we define below $\view_a(c)$ to be  the input  values of the agents in $\aview_a(c)$.

Consider for instance the simplicial model of Example~\ref{ex:binInputs} where three agents  $A=\{b, g, w\}$ each have a binary input value $0$ or $1$.
Let $M = \la C, \chi, \ell \ra$ be the corresponding simplicial model, and
denote a facet $X\in \cF(C)$ by a binary sequence $b_0 b_1 b_2$,
corresponding to the three values of $b, g, w$, in that order.

In the  \emph{immediate snapshot}  simplicial action model ${\mathcal{IS}}=\la T, \chi, \pre\ra$
for three agents $A=\{b,g,w\}$,  each action in $T$
is associated with a sequential partition of~$A$.
Furthermore, there is one copy of each sequential partition $c$ for
each facet $X\in \cF(C)$ of model $M$. 
Thus, an action of $T$ is given by the data $c,b_0,b_1,b_2$,
which we write $c^{b_0 b_1 b_2}$.
The precondition of the action $c^{b_0 b_1 b_2}$ is true precisely
in the facet $b_0 b_1 b_2$ of $M$ (it is a conjunction of atomic propositions).

Consider an action $c^{b_0 b_1 b_2}$, where 
$c=c_1,\ldots,c_m$. Then $c^{b_0 b_1 b_2}$ is interpreted as follows.
If an agent $a$ is in $c_j$
then $a$ learns the values (of facet $b_0 b_1 b_2$)
of all agents in $c_i$ for $i\leq j$, and only those values.
We write $\view_a(c^{b_0 b_1 b_2})$ the vector of values
that $a$ learned, namely, the vector obtained from
$b_0 b_1 b_2$ by replacing a value $b_i$ by $\emptyset$ for
agents not in  $\aview_a(c)$.
%As before, we write $\aview_a(c^{b_0 b_1 b_2})$ the set of agents seen by $a$,
%and let $\view_a(c^{b_0 b_1 b_2})$ be the vector of cards
%that $a$ learned, namely, the vector obtained from
%$b_0 b_1 b_2$ by replacing a value $b_i$ by $\emptyset$ for
%agents not in  $\aview_a(c^{b_0 b_1 b_2})$.
%
 Formally, the chromatic simplicial complex $\la T, \chi \ra$,
 consists of all facets of the form:
\[
 \{
 \la b, \view_b(c^{b_0 b_1 b_2})\ra,
 \la g, \view_g(c^{b_0 b_1 b_2})\ra,
 \la w, \view_w(c^{b_0 b_1 b_2})\ra
 \}
\]
%\begin{figure}[h]

\begin{center}
\includegraphics[scale=0.4]{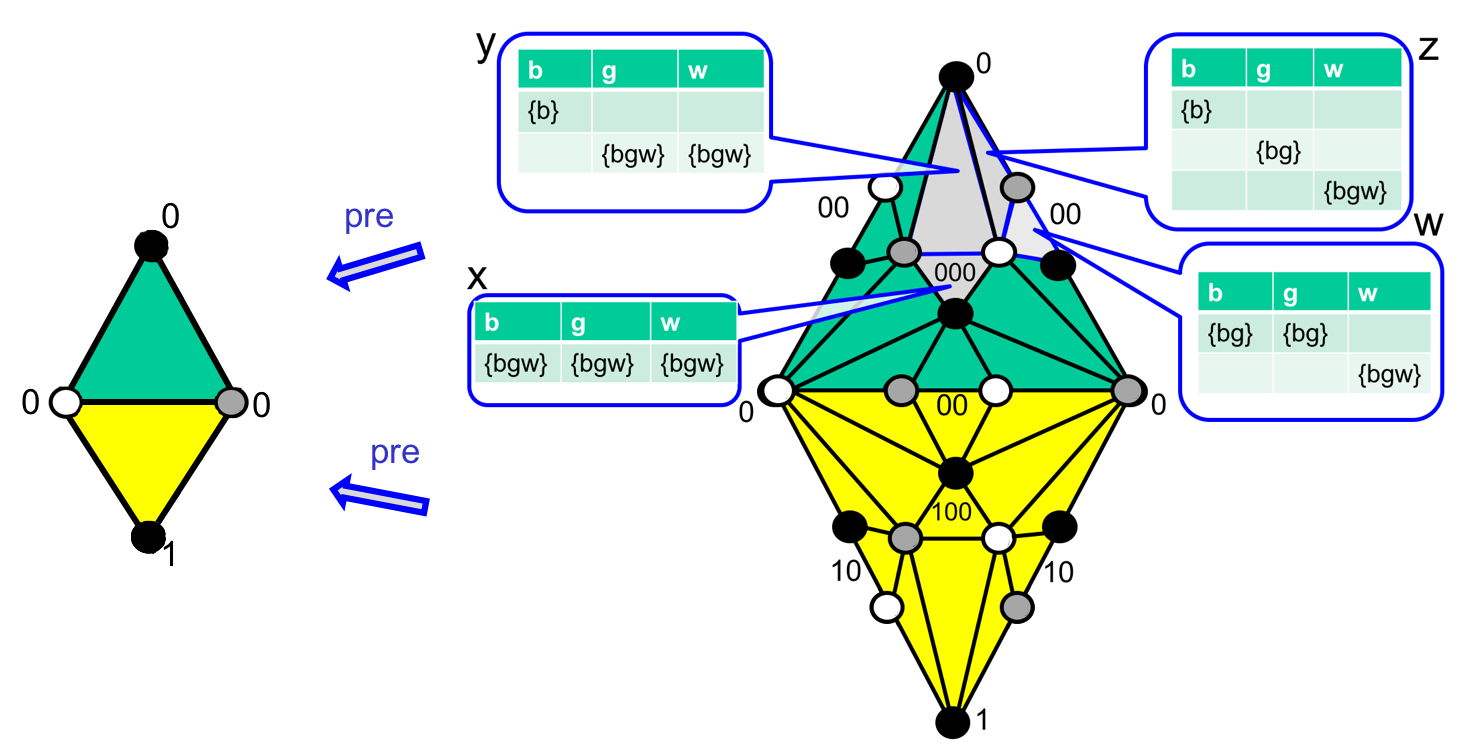}
%\includegraphics[scale=0.25]{sphere2.png}
%\caption{Kripke frame for 3 processes with binary inputs (not all edges depicted) and its dual complex.}
%\label{fig-spheres2}
\end{center}
%\end{figure}

The figure above illustrates (part of) the action model $\mathcal{IS}$.
It consists of the subdivisions of two triangles; the green copy above has
one triangle for each sequential partition (and
has four sequential partitions depicted). Similarly, the yellow
subdivided triangle below
repeats again all sequential partitions, but for a different precondition.
The precondition for all facets in the subdivided triangle above
is $000$, while for the facets of the subdivided triangle below it is $100$.
The subdivision  on top
has four facets identified, $X,Y,Z,W$, each one corresponding
to one of the four types of sequential partitions of ${A=\{b,w,g\}}$,
along with the corresponding $\aview$s shown in bubbles.
The colors black, grey, white of the vertices correspond respectively to agents $b$, $g$, $w$.
Notice that, for example, 
neither $b$ nor $w$ distinguish between actions $Y$ and~$Z$,
and indeed, their views are equal in $Y$ and $Z$: 
the view of $b$ consists of itself and the view of $w$
consists of the three inputs.
The numbers on the subdivided triangles indicate the $\view$s. 
 In the corners, an agent does not learn the input of any other
agent. In the boundary, two agents learn each other's inputs, and in
the center, all three learn each other's inputs.
%Finally, if the precondition for all facets in the subdivided triangle above
%is $000$, while for all the facets of the subdivided triangle below is $100$,
%then there is a boundary shared by both subdivisions, where 
%3 facets above intersect with 3 facets below.
Finally, let us look at what happens on the boundary shared by both subdivisions.
For example, the two facets in the middle of the figure
correspond to the sequential partition $\{gw\}\{b\}$;
neither $w$ nor $g$ have seen $b$, so they cannot tell whether the input of $b$ is $0$ or $1$.

\medskip
An action model is \emph{uniform} if its set of actions (facets) can be partitioned
into $k$ copies of a complex~$C$, called \emph{components},
such that all actions in $C_i$ have the same precondition,
which is true in exactly one facet $X_i$ of the simplicial model $M$.
The action model $\mathcal{IS}$ is indeed uniform, and its components
are isomorphic to a simplicial complex $C$, called the \emph{standard chromatic
subdivision},
that has been thoroughly studied. It is clear from the
figure that $C$ is a
subdivision, but for an arbitrary number of agents,
the proof is not simple~\cite{HerlihyKR:2013,kozlov2012}.  
It has been shown to have
several other topological properties, such as being collapsible~\cite{BenavidesR2018}. 
But in fact, for many applications such as consensus and set agreement,
it is sufficient to observe the following (see ch.9 of~\cite{HerlihyKR:2013}).

\begin{lemma}
\label{actioModIsSub}
Each  component of $\mathcal{IS}$ is 
a pseudomanifold with boundary.
If $M$ is a pseudomanifold with boundary, then so is $M[\mathcal{IS}]$.
\end{lemma}

%The precise definitions and the proof are given in Appendix~\ref{app:top}.
For a detailed proof see~\cite{AttiyaR2002}. 
To complete the example, 
notice that the effect of applying the action model $\mathcal{IS}$ to the model $M$ of Example~\ref{ex:binInputs}, which consists of a triangulated sphere,
is to subdivide each of the triangles
in the sphere.
Remarkably, the topology of the initial simplicial complex is preserved.
%For instance, such an action model could be applied to the triangulated
%torus  in Example~\ref{ex:torus}, to subdivide it further.
%Actions models that are {uniform} are of special interest,
%in the sense of the previous example, where the same set of actions
%can be applied to each facet of the initial model~$M$.

%\begin{theorem}
%\label{th:subdivAct}
%Consider a uniform action model $\cA$ whose components are
%chromatic subdivisions. Then $M[\cA]$ is a chromatic subdivision of $M$.
%\end{theorem}

In the $\mathcal{IS}$ model, each agent executes a single
immediate snapshot.
Iterating this model gives rise to the \emph{iterated immediate
snapshot model}~$\mathcal{IS}^r$~\cite{AttiyaR2002,iterated2010},
where each agent executes $r$ consecutive immediate snapshots.
Each component is a chromatic subdivision, where every triangle is
subdivided $r$ times.
%It can be generalized to an arbitrary number of immediate snapshots.
%, all performed on the same array~\cite{AttiyaR2002}, or each subsequent immediate snapshot by an agent is performed
%on a subsequent shared array~\cite{iterated2010}.
%Denote such an action model,
%where each agent executes $r$ immediate snapshots, by~$\mathcal{IS}^r$.
%One way of doing this, is by taking compositions of action models.
%Each component of $\mathcal{IS}^r$ is still a pseudomanifold 
%with boundary,
%in both the single array or the iterated array case,
%and in fact, each component is a  chromatic subdivision.

%%%%%%%%%%%%%%%%%%%%%%%%%%%%%%%%%%%%%%%%%%%%%%%%%%%

% !TEX root =  ./main.tex
%%%%%%%%%%%%%%%%%%%%%%%%%%%%%%%%%%%
\section{A DEL semantics for distributed task computability}
%\label{sec:kripkeModel}
\label{sec:modeKripkeF}
%We move to the orthogonal perspective, of considering sets of states instead of executions.
%, either initial states, or 
%states at the end of a set of executions. Each such set of states is used to define a Kripke frame.
%Once a model $M$ of computation and a protocol $D$ are fixed, one can state the question of weather a task
%is solvable by $D$ in $M$ in terms of Kripke frames.

\subsection{Tasks}
\label{sec:tasksolvability}
Consider the situation where a set of agents $A$
starts in an initial global state, defined by values given to each agent.
The values are local, in the sense that each agent knows its own initial value,
but not necessarily the values given to other agents. 
The agents communicate to each other their initial values,
via the immediate snapshot action model $\mathcal{IS}$ of Section~\ref{sec:basicAcM}.
Then, based on the information each agent has after communication,
the agent produces an output value. A task  specifies
 the output values that the agents may decide, when starting
in a given input state. Tasks have been studied since
early on in distributed computability~\cite{BiranMZ90}.
Here we provide, for the first time,  a DEL semantics for tasks.

Consider a simplicial model $\Gz = \la I, \chi, \ell \ra$ called the \emph{initial simplicial model}.
%consisting  of a pure chromatic simplicial 
%complex $\la I, \chi \ra$ for $n+1$ agents $A$, and a labeling $\ell : \cV(I) \to \mathscr{P}(\I{AP})$
%that associates to each vertex $v \in \cV(I)$ a set of atomic propositions that
%represent the initial knowledge of agent $\chi(v)$.
Each facet of $I$, with its labeling $\ell$, represents
a possible initial configuration.
Let us fix $\Gz$ to be the binary inputs 
model of Example~\ref{ex:binInputs}, to illustrate the ideas, and because it 
appears frequently in distributed computing. 
%For the pictures, let $A=\{b,g,w\}$.
%In this case,  all agents are in the same initial state, differing only in their
 %binary input values (implicitly they also know their names).
  
   A \emph{task} for $\Gz$ is a simplicial action model 
   ${\mathcal{T}}=\la T, \chi, \pre\ra$
for agents $A$, where each facet  
 is of the form $X=\{ \la b,d_b \ra, \la g,d_g \ra, \la w,d_w \ra\}$,
 where the values $d_b,d_g,d_w$ are taken from an arbitrary
 domain of \emph{output values}.
 Each such $X$ has 
a precondition that is true in one or more facets of $\Gz$,
interpreted as ``if the input configuration is a facet in which $\pre(X)$ holds,
and every agent $a \in A$ decides the value $d_a$, then this is a valid execution''.

The most important task in distributed computing is
\emph{binary consensus}, where the agents must agree on a value $0$ or $1$, such
that at least one agent had the agreed value as input.
Thus, $T$ has only two facets, $X_0$
where all decisions are $0$ and $X_1$, where all decisions are $1$.
$\pre(X_0)$ is true in all facets of $\Gz$, except for the one
where all agents start with input $1$. Similarly,
$\pre(X_1)$ is true in all facets of $\Gz$, except for the one
where all agents start with input $0$. 
%Intuitively, agents start
% with  value $0$ or $1$, and  they all have to the decide the same value,
% such that the value decided must be the input  of one of the agents.  
The following generalization of consensus has been well studied in distributed
computability \cite{HerlihyKR:2013}.
In the \emph{$k$-set agreement} task, agents start with inputs from a set of at least $k+1$
values, and have to decide on at most $k$ different inputs.

\subsection{Semantics of task solvability}
Given the  {simplicial input model}  $\Gz$
and a communication model $\cA$ such as $\mathcal{IS}$,
we get the \emph{simplicial protocol model} $\Gz[\cA]$,
that represents the knowledge gained by the agents % starting with $\Gz$, and
after  executing~$\cA$.
To solve a task ${\mathcal{T}}$, each agent, based on its own
knowledge, should produce an output value, such that the
vector of output values corresponds to a facet of~${\mathcal{T}}$,
respecting the preconditions of the task.

The  following gives a formal epistemic logic semantics to task solvability,
%where $\Gz$ is an  input model, $\cA$
%is a communication model and ${\mathcal{T}}$
%is a task.
Recall that
a {morphism} $\delta$ of simplicial models  is a chromatic simplicial map that preserves the labeling: $\ell'(f(v)) = \ell(v)$.
Also recall that the product update model $\Gz[\cA]$ is a sub-complex of the cartesian product $\Gz \times \cA$, whose vertices are of the form $(i,ac)$ with $i$ a vertex of $\Gz$ and $ac$ a vertex of $\cA$. We write $\pi_\Gz$ for the first projection on $\Gz$, which is a morphism of simplicial models.

\begin{definition}
\label{thm:Kripketasksolv2}
 A task ${\mathcal{T}}$ is \emph{solvable} in $\mathcal{A}$ 
  if there exists a morphism $\delta: \Gz[\cA]\rightarrow \Gz[{\mathcal{T}}]$ such that
$\pi_I\, \circ\, \delta=\pi_I$, i.e., the diagram of simplicial complexes below commutes.
\end{definition}

\begin{wrapfigure}[9]{r}{-0.4\textwidth}%   \centering
%\begin{wrapfigure}{r}{-0.4\textwidth}%   \centering
 \centering
\begin{tikzpicture}
  \node (s) {$\Gz[\cA]$};
  \node (xy) [below=2 of s] {$\Gz[\mathcal{T}]$};
  \node (x) [left=of xy] {$\Gz$};
  \draw[<-] (x) to node [sloped, above] {$\pi_\Gz$} (s);
  \draw[->, dashed, right] (s) to node {$\delta$} (xy);
  \draw[->] (xy) to node [below] {$\pi_\Gz$} (x);
\end{tikzpicture}\end{wrapfigure}
The justification for this definition is the following.
%The morphism $\delta$ takes a vertex $(i,ac)$ of $\Gz[\cA]$, colored by agent
%$a \in A$, to a vertex $(i,\la a,d_a \ra)$ of $\Gz[\cT]$ colored by the same
%agent $a$ (because morphisms preserve color).
%The fact that both vertices have the same first component~$i$ comes from the
%commutativity of the diagram with the first projection morphisms $\pi_\Gz$.
%This means that, in the situation where agent $a$ starts with input $i$, and the
%action $ac$ happens, agent $a$ decides value $d_a$.
%
A facet $X$  in $\Gz[\cA]$ corresponds to a pair $(i,ac)$,
where ${i \in \cF(\Gz)}$ represents input value assignments to all agents, 
and $ac\in\cF(\cA)$ represents an action,
codifying the communication exchanges that took place.
%where $F$ has the same literals attached to it of the input state $s$.
The morphism~$\delta$ takes~$X$ to a facet $\delta(X) = (i,dec)$ of $\Gz[\cT]$,
where $dec \in \cF(\cT)$ is the set of decision values that the agents will
choose in the situation $X$.
Moreover, $\pre(dec)$ holds in $i$, meaning that $dec$ corresponds to valid decision values for input $i$.
The commutativity of the diagram %with the first projection morphisms $\pi_\Gz$
expresses the fact that both $X$ and $\delta(X)$ correspond to the same input
assignment $i$.
Now consider a single vertex $v \in X$ with $\chi(v) = a \in A$. Then, agent $a$
decides its value solely according to its knowledge in $\Gz[\cA]$: if another
facet $X'$ contains $v$, then $\delta(v) \in \delta(X) \cap \delta(X')$, meaning
that $a$ has to decide the same value in both situations. 

%If $\delta(F) =\{ \la b,d_b\ra,\la w,d_w\ra,\la g,d_g \ra\}=dec$   each agent $i$ can decide (operationally, in its program) the value $d_i$, because the value is a function only of its knowledge in $\Gz[\cA]$ in $F$; more precisely,
%if $v\in F$ is $\chi(v)=i$, 
%and $F'$ is another facet containing $v$, then $\delta(F)\cap\delta(F')$
%contains the vertex $\la i,d_i\ra$.

%Second, these decisions  respect the task $\pre$ specification, because 
%if we consider an initial facet $s\in\cF(\Gz)$, then the 
%communication event starting in $s$ ends in a facet $s'=(s,ac)$, 
%which is then mapped to a facet $\delta(s')$ represented by 
%a pair $(s,dec)$, and hence $\pre(dec)$ holds in $s$.
%Notice that 
% $  L^{\I{AP}}(h(s)) =  L^{\I{AP}}(s) $.
%Finally, if such a morphism $\delta$ does not exist, then it is impossible to solve the task in $\cA$, because any such decisions would actually be defining the required morphism.

The diagram above  has two illuminating interpretations.
First, by Theorem~\ref{thm:lose-knowledge}, we know that the knowledge about the world of each  agent
can only decrease (or stay constant) along the   $\delta$ arrow.
 So agents  should improve knowledge  through 
communication, by going from $\cI$ to $\cI[\cA]$. 
The task is solvable if and only if there is enough knowledge
 in $\cI[\cA]$ to match the knowledge required by $\cI[\cT]$.
Secondly, the possibility of solving a task depends on the existence of
a certain simplicial map from the complex of $\cI[\cA]$ to the complex of
$\cI[\cT]$. Recall that a simplicial map is the discrete equivalent of
a continuous map, and hence task solvability is of a topological nature.
This leads us to the  connection with distributed computability described in this extended
abstract; 
further details are in~\cite{ericSergioDEL1-2017,ericSergioDEL2-2017}.
%Consider a set of processes communicating by writing and taking snapshots,
%where each process starts with a private input value, and after communicating
%with the others, produces an output value. 
%A task in this setting
%is defined by a set of input vectors, a set of output vectors, and
%an input/output relation $\Delta$. Solving the task means that,
% in an execution with input vector $I$, 
%the vector of outputs $O$ values produced by the processes
%belongs to $\Delta(I)$. The task is solvable in this sense
%if and only if it is solvable
%in the sense of Definition~\ref{thm:Kripketasksolv2}.

\subsection{Applications}

Here we describe how to use our DEL setting
to analyze solvability in the immediate-snapshot model of three well-studied 
distributed computing tasks:
consensus, set agreement, and approximate agreement. Their solvability is already
well-understood; our aim here is to understand the epistemic logic content of the known topological arguments that are used to prove unsolvability.

\paragraph{Consensus}
%Recall that each facet of an initial simplicial model
%$\Gz$ codifies an assignment of input values to the agents,
% $I=(in_0,\ldots,in_n)$, where the input value of agent $i$ is $in_i$,
% and the input values $in_i$ are taken from a domain of \emph{input values}.
Let  $\Gz=\la I,\chi,\ell\ra$ be the initial simplicial model for binary input values, 
and  ${\mathcal{T}}=\la T, \chi, \pre\ra$ be the action model for binary
consensus. 
Thus, $T$ has only two facets, $X_0$
where all decisions are $0$ and $X_1$, where all decisions are~$1$.
The underlying complex of $\Gz[\cT]$ consists of two disjoint simplicial complexes:
$I_0\times X_0$ and   $I_1\times X_1$,
where $I_0$ consists of all input facets with at least one $0$,
and
 $I_1$ consists of all input facets with at least one $1$.
 Notice that, in fact, each of the two complexes
 $I_i\times X_i$, for $i\in\{0,1\}$, is isomorphic to $I_i$, since~$X_i$  consists of just one facet.
%The valuation $\ell$ of a facet $X\times Y$ is just as it was in $X$,
%i.e. valuation at $X\times Y$ codifies the same input values appearing in $X$.

To show that binary consensus cannot be solved by the immediate snapshot protocol, we must prove that the map $\delta : \Gz[\cA] \to \Gz[\cT]$ of Definition~\ref{thm:Kripketasksolv2} does not exist.
The usual proof of impossibility uses a topological obstruction to the existence of $\delta$.
Here, instead, we exhibit a logical obstruction. 

\begin{theorem}
\label{th:consImp}
The binary consensus task
is not solvable by $\mathcal{IS}$.
\end{theorem} 

\begin{proof}
We first state some required knowledge at $\Gz[\cT ]$ to solve the task.
Let $\varphi_i$ be a formula denoting that at least one agent has input $i$.
%\begin{proposition}
%\label{consComKnInp}
We claim  that, for $i\in\{0,1\}$,
at any facet $Y$ of $I_i\times X_i$,
there is common knowledge  that at least one input is $i$,
$\Gz[\cT] ,Y \models C_A \varphi_i$.
%\end{proposition}

Now, consider the simplicial model $\Gz[\mathcal{IS}]$,
for the immediate snapshot action model.
By Lemma~\ref{actioModIsSub} % and Theorem~\ref{th:subdivAct}
the underlying complex of $\Gz[\mathcal{IS}]$ is strongly connected,
and hence there is a path from the facet with valuation indicating
that all inputs are $0$ to the facet where all inputs are $1$.
Namely, we claim that
%\begin{proposition}
%\label{consNoComKnInp}
at any facet $X$ of $\Gz[\mathcal{IS}]$, it is not the case
that 
$\Gz[\mathcal{IS}],X \models C_A \varphi_i$, for both $i=0$~and~$i=1$.
%\end{proposition}

Finally, we know that morphisms of simplicial models cannot ``gain knowledge about the world'' from
Theorem~\ref{thm:lose-knowledge}, and hence, there cannot be a morphism
$\delta$ from $\Gz[\mathcal{IS}]$ to $\Gz[\cT ]$, by the two previous
claims. 
\end{proof}
 
Two observations. First, notice that the proof argument holds
for any other model, instead of $\mathcal{IS}$, which is
connected. This is the case for any number
of communication rounds by wait-free asynchronous agents~\cite{Herlihy:waitFree1988},
and even if only one may crash in a message passing system~\cite{FischerLP85}, which are the classic consensus impossibility results.
Secondly, the usual topological argument for impossibility is the following: because simplicial maps preserve 
connectivity, $\delta$ cannot send a connected
simplicial complex into a disconnected simplicial complex.
Notice how in both the logical and the topological proofs, the main ingredient is a connectedness argument.

\paragraph{Set agreement}
Let  $\Gz=\la I,\chi,\ell\ra$ be the initial simplicial model for $A=\{b,w,g\}$,
and three possible input values, $\{0,1,2\}$.
Let  ${\mathcal{T}}=\la T, \chi, \pre\ra$ be the action model for $2$-set agreement, requiring that each agent decides on one of the input
values, and at most $2$ different values are decided.
Thus, $T$ has  facets $X_{d_0,d_1,d_2}$, for each vector $d_0,d_1,d_2$,
such that $d_i \in \{0,1,2\}$, $|\{d_0, d_1, d_2\}| \leq 2$,
and $\pre(X_{d_0,d_1,d_2}) = \varphi_{d_0} \land \varphi_{d_1} \land \varphi_{d_2}$,
where $\varphi_i$ is as above.

\begin{theorem}
\label{th:saImp}
The $2$-set agreement task
is not solvable by $\mathcal{IS}$.
\end{theorem} 

\begin{proof}[Proof (Sketch).]
The usual topological argument~\cite{HerlihyKR:2013} roughly goes as follows.
%
%In $\Gz[\mathcal{IS}]$, we have the following triangular structure.
%The three ``corners'' of the triangle, for $i \in \{0,1,2\}$ correspond to the
%three input facets $I_i$ where all agents have input $i$.
%
%By Lemma~\ref{actioModIsSub}, the complex $\Gz[\mathcal{IS}]$ is a pseudomanifold.
%
We can visualize the complex $\Gz[\cT]$
as having the structure of a triangle with a hole in the middle.
The three "corners" of the triangle, indexed by $i\in\{ 0,1,2\}$, consist
of just one facet of the form  ${I}_i\times X_{i}$,
where ${I}_i$ is the input facet with only one input, $i$,
and in $X_{i}$ all decisions are~$i$.
The three "edges" of the triangle are of the form ${I}_{ij}\times X_{ij}$,
where $I_{ij}$ consists of all facets with inputs in $\{ i,j\}$,
and $X_{ij}$ the facets whose decisions are in $\{i,j\}$.
Notice that ${I}_{ij}\times X_{ij}$ contains ${I}_{i}\times X_i$ and $I_j \times X_j$.
This triangle must have a hole in the middle: otherwise, by Sperner's lemma (see
e.g. \cite{kozlov:2007,HerlihyKR:2013}), %Appendix~\ref{app:sperner}), 
there would be a facet with three distinct
decision values.
Thus, $\Gz[\cT]$ is not $2$-connected. But since $\Gz[\mathcal{IS}]$ is $2$-connected,
and simplicial maps preserve $2$-connectivity, there cannot exist a suitable
$\delta$.
\end{proof}

Finding a logical obstruction to the existence of $\delta$ is an open question.
This would amount to finding a formula $\varphi$ which is true
$\Gz[\cT]$, but false $\Gz[\mathcal{IS}]$, and applying
Theorem~\ref{thm:lose-knowledge}. Doing so, we would understand better what
knowledge is necessary to solve set agreement, which is not available in
$\Gz[\mathcal{IS}]$.

\paragraph{Approximate agreement}
We discuss now the approximate agreement
task, where agents have to decide values which are $1/N$ apart
from each other.
Its solvability depends on the number
of immediate snapshot communication rounds $r$ that the
agents perform. We did not describe in detail
the action model $\mathcal{IS}^r$, 
so we only briefly mention that the task is solvable
if and only if the number of rounds is large enough, with
respect to $N$.
Very roughly, there is a center facet $X^c$ in the
product update of the task, where
$\Gz[\cT],X^c \models E^k \phi_c$, where  $k$
is roughly $N/2$, 
for a formula $\phi_c$ representing the input values in $X^c$.
On the other hand, 
there is no facet $X$ in $\Gz[\mathcal{IS}^r]$ where this
knowledge exists, unless $r$ is large enough.
A detailed proof is in the technical report~\cite{ericSergioDEL1-2017}.

\section{Conclusions}
\label{sec:conclusions}
We have made a first step into defining a version of fault-tolerant multi-agent DEL
using simplicial complexes as models, providing a different perspective from
the classical knowledge approach based on Kripke frames.
Also, we have defined problem specifications based on DEL
using simplical complexes, instead of based on 
formula specifications.
We have thus established a bridge between 
the theory of distributed computability and epistemic logic.
We illustrated the setting with a simple one-round communication action
model $\mathcal{IS}$, that corresponds
to a well-studied situation in distributed computing, but many
other models can be treated similarly.
%, including multi-round communication versions
%of $\mathcal{IS}$.
%In this way we have provided a formal semantics to this distributed
%computing model, in one direction, and in the other direction,
%it is possible to state the knowledge needed to solve a task.

Many interesting questions are left for future work.
We have developed here all our theory on pure simplicial complexes,
where all the facets are of the same dimension.
Extending it to complexes with lower dimensional facets would allow
us to model detectable failures.
In two preliminary reports we give additional details,
and explore further some of these issues~\cite{ericSergioDEL1-2017,ericSergioDEL2-2017}.
It is known to be undecidable whether a task is solvable in the immediate
snapshot model, even for three processes~\cite{GKdec1999,HRdec1997},
and hence the connection we establish with DEL implies
that it is undecidable if certain knowledge has been gained in multi-round
immediate snapshot action models, but further work is needed to
study this issue.
%Another, main avenue that we leave for 
Future work is needed to study bisimulations and their relation to
%which are very important in DEL.
the simulations 
% between various  models have  been 
studied in task computability~\cite{HerlihyR12}. 
%It would be interesting to study such simulations from the bi-similarity perspective in epistemic models.
%Compositionality has been studied in DEL, and also in distributed computing,
%it would be interesting to explore links in this direction.
It would be of interest to study other distributed computing settings, especially those which have stronger
communication objects available, and which are known to yield complexes
that might not preserve the topology of the input complex.
%; their additional power is expressed as higher dimensional ``holes''.
%It would be very interesting to find a formalization of such topological properties
%in terms of knowledge, and thus obtain a generalization from common knowledge (that is tightly related to 1-dimensional connectivity)
%to other form of group knowledge (related to higher-dimensional connectivity).

%%%%%%%%%%%%%%%%%%%%%%%%%%%%%%%%%%%%%%%%%%%%%%%%%%%

\comment{
\paragraph*{Acknowledgements.}
We thank Carlos Velarde and David Rosenblueth for their involvement
in the early stages of this research, and their help in developing the dual of a Kripke graph.
Part of this work was done while Sergio Rajsbaum was at Ecole Polytechnique.
This work was partially supported by PAPIIT-UNAM IN109917.
}

%%%%%%%%%%%%%%%%%%%%%%%%%%%%%%%%%%%%%%%%%%%%%%%%%%%

\paragraph{Acknowledgments}
This work has been partially supported by UNAM-PAPIIT IN109917
and France-Mexico ECOS 207560(M12-M01).
%\nocite{*}
\bibliographystyle{eptcs}
\bibliography{logic}

%\newpage
%\appendix
%\input{app.tex}

\end{document}